\documentclass[submission,copyright,creativecommons]{eptcs}

\usepackage{negotiations}

\usepackage{graphicx}
\usepackage{listings}

\usepackage{rotating}

\usepackage{enumerate}

\usepackage{amsthm}
\newtheorem{theorem}{Theorem}[section]
\newtheorem{definition}[theorem]{Definition}

\newtheorem{lemma}[theorem]{Lemma}
\newtheorem{corollary}[theorem]{Corollary}

\newtheorem{theoremDummy}{Theorem}
\newenvironment{reftheorem}[1]
{\renewcommand{\thetheoremDummy}{\ref{#1}}\begin{theoremDummy}}
{\end{theoremDummy}}

\newtheorem{corollaryDummy}{Theorem}
\newenvironment{refcorollary}[1]
{\renewcommand{\thetheoremDummy}{\ref{#1}}\begin{corollaryDummy}}
{\end{corollaryDummy}}

\begin{document}

\title{Negotiation Games}
\author{Philipp Hoffmann
\institute{Fakult\"at f\"{u}r Informatik\\
Technische Universit\"{a}t M\"{u}nchen\\
Germany}
\email{ph.hoffmann@tum.de}
}
\newcommand{\titlerunning}{Negotiation Games}
\newcommand{\authorrunning}{Philipp Hoffmann}
\maketitle

\begin{abstract}
Negotiations, a model of concurrency with multi party negotiation as primitive,
have been recently introduced in \cite{negI,negII}. We initiate the study of games for 
this model. We study {\em coalition problems}: can a given coalition of agents 
force that a negotiation terminates (resp. block the negotiation so that it goes 
on forever)?; can the coalition force a given outcome of the negotiation? We show
that for arbitrary negotiations the problems are EXPTIME-complete. 
Then we show that for sound and deterministic or even weakly deterministic
negotiations the problems can be solved in PTIME. Notice that the 
input of the problems is a negotiation, which can be exponentially more compact 
than its state space.
\end{abstract}

\section{Introduction}

In \cite{negI,negII}, the first author and J\"org Desel have introduced 
a model of concurrency with multi party negotiation as primitive. The model
allows one to describe distributed negotiations obtained by combining
``atomic'' multi party negotiations, or \emph{atoms}. Each atom has a number 
of \emph{parties} (the subset of agents involved), and a set  
of possible outcomes. The parties agree on an outcome, which
determines for each party the subset of atoms 
it is ready to engage in next.

Ill-designed negotiations may deadlock, or may contains useless 
atoms, i.e., atoms that can never be executed. The problem whether a 
negotiation is well designed or \emph{sound} was studied in \cite{negI,negII}. 
The main result was the identification of two classes,
called deterministic and acyclic weakly deterministic negotiations, for which the 
soundness problem is tractable: while the problem is PSPACE-complete for arbitrary 
negotiations, it becomes polynomial for these two classes.

In this paper we start the study of games on negotiations. As for games
played on pushdown automata \cite{DBLP:journals/iandc/Walukiewicz01},
vector addition systems with states (VASS) \cite{DBLP:conf/icalp/BrazdilJK10}, counter
machines \cite{DBLP:conf/rp/Kucera12}, or
asynchronous automata \cite{DBLP:conf/fsttcs/MohalikW03}, games on negotiations
can be translated into games played on the (reachable part of the) 
state space. However, the number of states of a negotiation may grow exponentially 
in the number of agents, and so the state space can be exponentially 
larger than the negotiation. We explore the complexity of solving games 
\emph{in the size of the negotiation}, not on the size of the state space. 
In particular, we are interested in finding negotiation classes for which the 
winner can be decided in polynomial time, thus solving the state space explosion problem. 

We study games formalizing the two most interesting questions related to a 
negotiation. First, can a given \emph{coalition} (i.e., a given subset of agents) 
force termination of the negotiation? (Negotiations may contain cycles.) Second, 
can the coalition force a given final outcome?

Our first results show that these two problems are EXPTIME-complete in the 
size of the negotiation. This is the case even if the negotiation is 
deterministic, and so it seems as if the tractability results of 
\cite{negI,negII} cannot be extended to games. But then, we are able to show that, very 
surprisingly, the problems are polynomial for deterministic (or even weakly 
deterministic) negotiations \emph{that are sound}. This is very satisfactory: 
since unsound negotiations are ill-designed, we are not interested in them anyway.
And, very unexpectedly, the restriction to sound negotiations has as collateral effect a
dramatic improvement in the complexity of the problem. Moreover, 
the restriction comes ``at no cost'', because deciding soundness
of deterministic negotiations is also decidable in polynomial time.

\paragraph*{Related work.}
Our games can be seen as special cases of concurrent games 
\cite{DBLP:journals/jacm/AlurHK02,DBLP:journals/tcs/AlfaroHK07} in which 
the arena is succinctly represented as a negotiation. Explicit construction of the 
arena and application of the algorithms of 
\cite{DBLP:journals/jacm/AlurHK02,DBLP:journals/tcs/AlfaroHK07} yields an exponential algorithm, 
while we provide a polynomial one. 

Negotiations have the same expressive power as 1-safe Petri nets or 1-safe VASS, 
although they can be exponentially more compact (see \cite{negI,negII}). Games for
unrestricted VASS have been studied in \cite{DBLP:conf/icalp/BrazdilJK10}. 
However, in \cite{DBLP:conf/icalp/BrazdilJK10} the emphasis is on VASS with an 
infinite state space, while we concentrate on the 1-safe case.

The papers closer to ours are those studying games on asynchronous automata
(see e.g. \cite{DBLP:conf/fsttcs/MohalikW03,DBLP:journals/fmsd/GastinSZ09,DBLP:conf/icalp/GenestGMW13}). Like negotiations, asynchronous automata are a model of distributed computation
with a finite state space. These papers study algorithms for deciding the existence of 
distributed strategies for a game, i.e., local strategies for each agent based only on 
the information the agent has on the global system. Our results identify a special case 
with much lower complexity than the general one, in which local strategies are even memoryless.

Finally, economists have studied mathematical models of negotiation games,
but with different goals and techniques (see e.g. \cite{rubinstein1982perfect}). 
In our terminology, they typically consider negotiations 
in which all agents participate in all atomic negotiations. We focus on distributed
negotiations, where in particular atomic negotiations involving disjoint sets
of agents may occur concurrently.

\section{Negotiations: Syntax and Semantics}
Negotiations are introduced in \cite{negI}. We recall the main definitions.
We fix a finite set $\agents$ of \emph{agents} representing potential parties 
of negotiations. In \cite{negI,negII} each agent has an associated set of 
internal states. For the purpose of this paper the internal states are irrelevant, and so we omit them. 

\smallskip
\noindent \textbf{Atoms.} A \emph{negotiation atom}, or just an \emph{atom}, is a pair
$n=(P_n, \outc_n)$, where $P_n \subseteq \agents$ is a nonempty set of \emph{parties}, 
and $\outc_n$ is a finite, nonempty set of \emph{outcomes}.

\smallskip
\noindent \textbf{(Distributed) Negotiations.} A distributed negotiation is
a set of atoms together with a \emph{transition function} $\trans$ that assigns to every 
triple $(n,a,r)$ consisting of an atom $n$, a party $a$ of $n$, and an outcome $r$ of $n$ a set 
$\trans(n,a,r)$ of atoms. Intuitively, this is the set of atomic negotiations 
agent $a$ is ready to engage in after the atom $n$, if the outcome of $n$ is $r$. 

Formally, given a finite set of atoms $N$, let $T(N)$ denote the set of triples $(n, a, r)$ such that $n \in N$, $a\in P_n$, and $r \in \outc_n$. 
A \emph{negotiation} is a tuple $\N=(N, n_0, n_f, \trans)$, where 
$n_0, n_f \in N$ are the \emph{initial} and \emph{final} atoms, and  
$\trans \colon T(N) \rightarrow 2^N$ is the \emph{transition function}. Further, 
$\N$ satisfies the following properties: 
(1) every agent of $\agents$ participates in both $n_0$ and $n_f$; 
(2) for every $(n, a, r) \in T(N)$: $\trans(n, a, r)= \emptyset$ if{}f $n=n_f$.

\begin{figure}[bth]
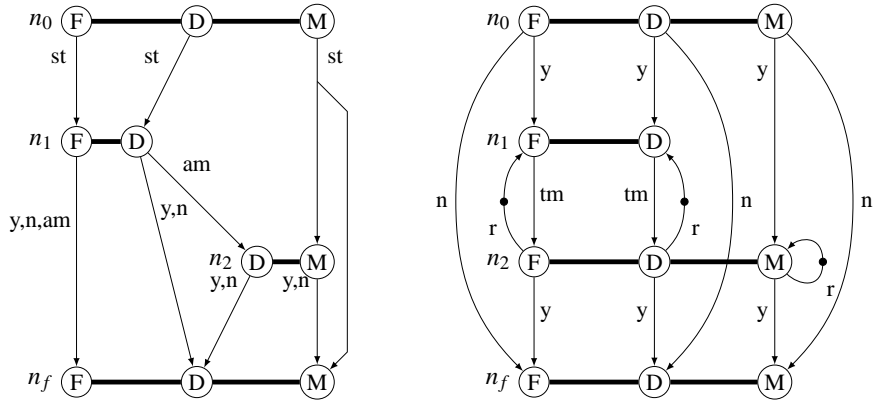

\centering
\scalebox{0.8}{
\input{tikz/FDM_acyclic}
\hspace{1cm}
\input{tikz/FDM_cyclic}
}
\caption{An acyclic and a cyclic negotiation.}
\label{fig:dneg}
\end{figure}

\noindent \textbf{Graphical representation.} Negotiations are graphically represented as shown in Figure \ref{fig:dneg}. 
For each atom $n \in N$ we draw a black bar; for each party $a$ of $P_n$ we 
draw a white circle on the bar, called a \emph{port}. For each $(n,a,r) \in T(N)$, 
we draw a hyper-arc leading from the port of $a$ in $n$ to all the ports of $a$ in 
the atoms of $\trans(n,a,r)$, and label it by $r$. 
Figure \ref{fig:dneg} shows two Father-Daughter-Mother negotiations. On the left,
Daughter and Father negotiate with possible outcomes \texttt{yes} ($\texttt{y}$), 
\texttt{no} ($\texttt{n}$),
and \texttt{ask\_mother} ($\texttt{am}$). If the outcome is the latter, then Daughter and Mother negotiate
with outcomes \texttt{yes}, \texttt{no}. In the negotiation on the right,
Father, Daughter and Mother negotiate with outcomes \texttt{yes} and \texttt{no}.
If the outcome is \texttt{yes}, then Father and Daughter negotiate a return time  
(atom $n_1$) and propose it to Mother (atom $n_2$). 
If Mother approves (outcome \texttt{yes}), then the negotiation terminates, 
otherwise (outcome \texttt{r}) Daughter and Father renegotiate the return time. 

\smallskip
\noindent \textbf{Semantics. } A \emph{marking} of a negotiation $\N=(N, n_0, n_f, \trans)$ is a mapping 
$\vx \colon \agents \rightarrow 2^N$. Intuitively, $\vx(a)$ is the set of atoms that agent $a$ is currently ready to engage in next. 
The \emph{initial} and \emph{final} markings,  denoted by $\vx_0$ and $\vx_f$ respectively, are given by $\vx_0(a)=\{n_0\}$ and 
$\vx_f(a)=\emptyset$ for every $a \in \agents$. 

A marking $\vx$ \emph{enables} an atom $n$ if $n \in \vx(a)$ for every $a \in P_n$,
i.e., if every party of $n$ is currently ready to engage in it.
If $\vx$ enables $n$, then $n$ can take place and its parties
agree on an outcome $r$; we say that $(n,r)$ \emph{occurs}.
Abusing language, we will call this pair also an outcome.
The occurrence of $(n,r)$ produces a next marking $\vx'$ given by $\vx'(a) = \trans(n,a,r)$ for every $a \in P_n$, 
and $\vx'(a)=\vx(a)$ for every $a \in \agents \setminus P_n$. 
We write $\vx \by{(n,r)} \vx'$ to denote this. 

By this definition, $\vx (a)$ is always either $\{n_0\}$ or equals 
$\trans(n,a,r)$ for some atom $n$ and outcome $r$. 
The marking $\vx_f$ can only be reached by the occurrence of 
$(n_f, r)$ ($r$ being a possible outcome of $n_f$), 
and it does not enable any atom. Any other marking that does not enable 
any atom is a \emph{deadlock}.

Reachable markings are graphically represented by placing tokens (black dots) on the forking points of the hyper-arcs (or in the middle of an arc). Figure \ref{fig:dneg} shows on the right a marking in which \texttt{F} and \texttt{D} are ready to engaging $n_1$ and \texttt{M} is ready to engage in $n_2$.

We write $\vx_1 \by{\sigma}$ to denote that there is a sequence 
\begin{equation*}
\vx_1 \by{(n_1,r_1)} \vx_2 \by{(n_2,r_2)}\cdots \by{(n_{k-1},r_{k-1})} \vx_{k} \by{(n_k,r_k)} \vx_{k+1} \cdots
\end{equation*}
such that  $\sigma = (n_1, r_1) \ldots (n_{k}, r_{k}) \ldots$. If
$\vx_1 \by{\sigma}$, then $\sigma$ is an \emph{occurrence sequence} from the marking $\vx_1$, and $\vx_1$ enables $\sigma$.
If $\sigma$ is finite, then we write
$\vx_1 \by{\sigma} \vx_{k+1}$ and say that $\vx_{k+1}$ is \emph{reachable} from $\vx_1$. 

\smallskip
\noindent \textbf{Soundness. } A negotiation is \emph{sound} if \emph{(a)} every atom is enabled at some reachable marking, and \emph{(b)} every occurrence sequence from the initial marking 
either leads to the final marking $\vx_f$, or can be extended to an
occurrence sequence that leads to $\vx_f$.  

The negotiations of Figure \ref{fig:dneg} are sound. However, if we set in the left negotiation
$\trans(n_0,\texttt{M}, \texttt{st})= \{n_2\}$ instead of $\trans(n_0,\texttt{M}, \texttt{st})= \{n_2, n_f\}$, then the occurrence sequence $(n_0,\texttt{st}) (n_1, \texttt{yes})$
leads to a deadlock.

\smallskip
\noindent \textbf{Determinism and weak determinism. } An agent $a \in \agents$ is \emph{deterministic} if for every $(n,a,r) \in T(N)$ such that $n \neq n_f$ there exists an atom
$n'$ such that $\trans(n,a,r) = \{n'\}$. 
    
The negotiation $\N$  is \emph{weakly deterministic} if for every $(n,a,r) \in T(N)$ there is a deterministic
agent $b$ that is a party of every atom in $\trans(n,a,r)$, i.e., $b \in P_{n'}$ for every $n' \in \trans(n,a,r)$.
In particular, every reachable atom has a deterministic party.
It is \emph{deterministic} if all its agents are deterministic.

Graphically, an agent $a$ is deterministic if no proper hyper-arc leaves any port 
of $a$, and a negotiation is deterministic if there are no proper hyper-arcs.
The negotiation on the left of Figure \ref{fig:dneg} is not deterministic (it contains a proper hyper-arc for Mother), while the one on the right is deterministic. 

\section{Games on Negotiations}
\label{sec:games}

We study a setting that includes, as a special case, the questions about coalitions mentioned
in the introduction: Can a given coalition (subset of agents) force termination of the negotiation?
Can the coalition force a given concluding outcome?

In many negotiations, there are reachable markings that enable more than one atom. If two of those atoms share an agent, the occurrence of one might disable the other and they are not truly concurrent. For a game where we want to allow concurrent moves, we formalize this concept with the notion of an independent set of atoms.

\begin{definition}
A set of atoms $S$ is independent if no two distinct atoms of $S$ share an agent, i.e., $P_n \cap P_{n'} = \emptyset$
for every $n, n' \in S$, $n \neq n'$.
\end{definition}
It follows immediately from the semantics that if a marking $\vx$ enables all atoms of $S$ and we fix an outcome $\r_i$ for each $n_i \in S$,
then there is a unique marking $\vx'$ such that $\vx \by{\sigma} \vx'$ for every 
sequence $\sigma = (n_1,\r_1) \ldots (n_k,\r_k)$ such that each atom of $S$ appears exactly once in
$\sigma$. In other words, $\vx'$ depends only on the outcomes of the atoms, and not on the order in
which they occur.

A \emph{negotiation arena} is a negotiation whose set $N$ of atoms 
is partitioned into two sets $N_1$ and $N_2$. We consider  
concurrent games \cite{DBLP:journals/jacm/AlurHK02,DBLP:journals/tcs/AlfaroHK07} with three 
players called Player 1, Player 2, and Scheduler. At each step, Scheduler chooses a nonempty set
of independent atoms among the atoms enabled at the current marking of the negotiation arena. 
Then, Player 1 and Player 2, independently of each other, select an outcome 
for each atom in $S\cap N_1$ and $S\cap N_2$, respectively. Finally, 
these outcomes occur in any order, and the game moves to the unique marking 
$\vx'$ mentioned above. The game terminates if it reaches a marking enabling no atoms,
otherwise it continues forever. 

Formally, a partial play is a sequence of tuples $(S_i, F_{i,1}, F_{i,2})$ 
where each $S_i \subseteq N$ is a set of independent atoms and $F_{i,j}$ assigns to every 
$n \in S_i\cap N_j$ an outcome $\r \in R_n$. Furthermore it must hold that every atom $n \in S_i$ 
is enabled after all atoms in $S_0,...,S_{i-1}$ have occurred with the outcomes specified by $F_{0,1}, F_{0,2},...,F_{i-1,1}, F_{i-1,2}$. 
A play is a partial play that is either infinite or reaches a marking enabling no atoms.
For a play $\pi$ we denote by $\pi_i$ the partial play consisting of the first $i$ tuples of $\pi$.

We consider two different winning conditions. In the \emph{termination game}, 
Player 1 wins a play if the play ends with $n_f$ occurring, otherwise Player 2 wins. In the \emph{concluding-outcome game},
we select for each agent $a$ a set of outcomes $G_a$ such that $n_f \in \trans(n,a,r)$ for $r \in G_a$ (that is, after any outcome $\r \in G_a$, agent $a$ is ready to terminate). Player 1 wins if the the play ends with $n_f$ occurring, and for each agent $a$ the last outcome $(n, \r)$ of the play such that $a$ is a party of $n$ belongs to $G_a$. 

A strategy $\sigma$ for Player $j, j\in\{1, 2\}$ is a partial function that, given a partial play 
$\pi = (S_0, F_{0,1},F_{0,2}),..,(S_i, F_{i,1},F_{i,2})$ and a set of atoms $S_{i+1}$ returns 
a function $F_{i+1,j}$ 
according to the constraints above. A play $\pi$ is said to be played according to a 
strategy $\sigma$ of Player $j$ 
if for all $i$, $\sigma(\pi_i, S_{i+1}) = F_{i+1,j}$. A strategy $\sigma$ is a winning 
strategy for Player $j$ if he 
wins every play that is played according to $\sigma$. Player $j$ is said to win
 the game if he has a winning strategy. Notice that if Player 1 has a winning strategy 
then he wins every play against any pair of strategies for Player 2 and Scheduler. 

\begin{definition}
Let $\N$ be a negotiation arena.
The termination (resp. concluding-outcome) problem for $\N$ consists of deciding whether Player 1 has a winning strategy for the termination game (concluding-outcome game).
\end{definition}

\begin{figure}[h]
\centering
\scalebox{1.1}{
\input{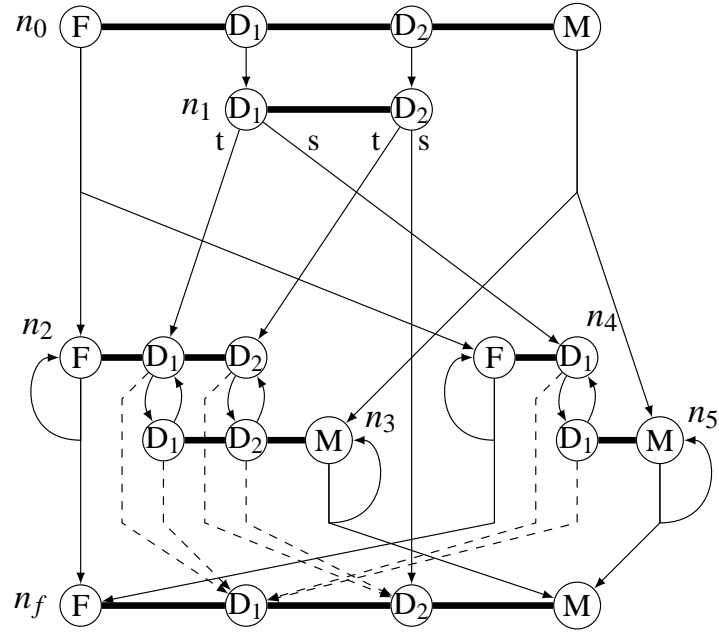}
}
\caption{Atom control and determinism}
\label{fddm}
\end{figure}

Assume we want to model the following situation: In a family with Father (F), Mother (M) and two Daughters (D$_1$ and D$_2$), Daughter D$_1$ wants to go to a party. She can talk to each parent individually, but can choose beforehand whether to take her sister D$_2$ with her or not. Figure \ref{fddm} models this negotiation. 
The solid edges for the daughters between $n_2$ and $n_3$ and
between $n_4$ and $n_5$ ``ask the other parent'' outcomes, while the dashed edges 
represent the ``yes'' and ``no'' outcomes. Assume the daughters work together to reach termination.
Then in $N_1 = \{ n_1, n_2, n_3 \}$ the daughters have 
a majority which we will interpret as ``they can choose which outcome is taken''. Can the daughters force termination?

At atom $n_1$ the daughters decide whether D$_2$ should participate in the conversation with the parents (outcome \texttt{t}) or not (outcome \texttt{s}). 
If the daughters choose outcome \texttt{s},
then Father and Mother can force an infinite loop between $n_4$ and $n_5$. On the contrary, if the daughters choose to stay together, then, 
since they control atom $n_2$, they can force a  ``yes'' or ``no'' outcome, and therefore termination. 

The questions whether a coalition ${\cal C}$ of agents can force termination or a 
certain outcome are special instances of the termination and concluding-outcome problems. 
In these instances, an atom $n$ belongs to $N_1$---the set 
of atoms controlled by Player 1--- if{}f a strict majority of the agents of $n$ 
are members of ${\cal C}$. 

\subsection{Coalitions}

Before we turn to the termination and concluding-outcome problems, we briefly study coalitions. Intuitively, a coalition controls all the atoms where it has strict majority. We show that while the definition of the partition of the atoms $N$ according to the participating agents may seem restricting, this is not the case: In all cases but the deterministic sound case, any partition can be reached, possibly by adding agents.

We define the partition of $N$ via a partition of the agents: Let the agents $A$ be partitioned into two sets $A_1$ and $A_2$. Define $N_1 = \{ n \in N : |P_n \cap A_1| > |P_n \cap A_2|\}$, $N_2 = N \backslash N_1$. Note that ties are controlled by $A_2$.

We first show that in the nondeterministic and weakly deterministic case, this definition is equivalent to one where we decide control for each atom and not for each agent.

\begin{figure}[h]
\centering
\scalebox{0.8}{
\input{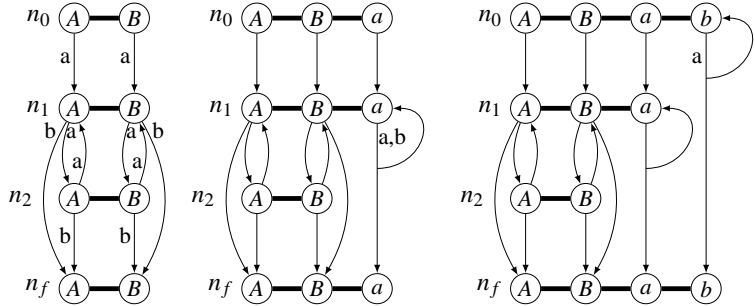}
\tikz{
\nego[text={$A$}{$B$}{$a$},id=n0]{0,0}
\node[left = 0cm of n0_P0, font=\large] {$n_0$};
\nego[text={$A$}{$B$}{$a$},id=n1]{0,-1.5}
\node[left = 0cm of n1_P0, font=\large] {$n_1$};
\nego[text={$A$}{$B$},id=n2]{0,-3}
\node[left = 0.3cm of n2_P0, font=\large] {$n_2$};
\nego[text={$A$}{$B$}{$a$},id=nf]{0,-4.5}
\node[left = 0cm of nf_P0, font=\large] {$n_f$};
\pgfsetarrowsend{latex}
\draw (n0_P0) -- (n1_P0);
\draw (n0_P1) -- (n1_P1);
\draw (n0_P2) -- (n1_P2);
\draw (n1_P0) to[out=-110,in=110] (n2_P0);
\draw (n1_P1) to[out=-110,in=110] (n2_P1);
\draw (n1_P0) to[out=-120,in=120] (nf_P0);
\draw (n1_P1) to[out=-60,in=60] (nf_P1);
\draw (n1_P2) -- (nf_P2);
\draw (n1_P2) + (0,-1) to [out=0,in=0,out looseness = 2,in looseness=2] (n1_P2);
\draw (n2_P0) to[out=70,in=-70] (n1_P0);
\draw (n2_P1) to[out=70,in=-70] (n1_P1);
\draw (n2_P0) -- (nf_P0);
\draw (n2_P1) -- (nf_P1);
\labelEdge[start=n1_P2,end=nf_P2,dist=0.65,shift={0.25,0}]{a,b}
}
\tikz{
\nego[text={$A$}{$B$}{$a$}{$b$},id=n0]{0,0}
\node[left = 0cm of n0_P0, font=\large] {$n_0$};
\nego[text={$A$}{$B$}{$a$},id=n1]{0,-1.5}
\node[left = 0cm of n1_P0, font=\large] {$n_1$};
\nego[text={$A$}{$B$},id=n2]{0,-3}
\node[left = 0.3cm of n2_P0, font=\large] {$n_2$};
\nego[text={$A$}{$B$}{$a$}{$b$},id=nf]{0,-4.5}
\node[left = 0cm of nf_P0, font=\large] {$n_f$};
\pgfsetarrowsend{latex}
\draw (n0_P0) -- (n1_P0);
\draw (n0_P1) -- (n1_P1);
\draw (n0_P2) -- (n1_P2);
\draw (n0_P3) -- (nf_P3);
\draw (n0_P3) + (0,-1) to [out=0,in=0,out looseness = 2,in looseness=2] (n0_P3);
\draw (n1_P0) to[out=-110,in=110] (n2_P0);
\draw (n1_P1) to[out=-110,in=110] (n2_P1);
\draw (n1_P0) to[out=-120,in=120] (nf_P0);
\draw (n1_P1) to[out=-60,in=60] (nf_P1);
\draw (n1_P2) -- (nf_P2);
\draw (n1_P2) + (0,-1) to [out=0,in=0,out looseness = 2,in looseness=2] (n1_P2);
\draw (n2_P0) to[out=70,in=-70] (n1_P0);
\draw (n2_P1) to[out=70,in=-70] (n1_P1);
\draw (n2_P0) -- (nf_P0);
\draw (n2_P1) -- (nf_P1);
\labelEdge[start=n0_P3,end=nf_P3,dist=0.75,shift={-0.15,0}]{a}
}
}
\caption{Atom control via additional agents}
\label{fig:control}
\end{figure}

Consider the example given in Figure \ref{fig:control}. On the left a deterministic negotiation with two agents is given. Assume the coalitions are $A_1 = \{A\}$ and $A_2 = \{B\}$. By the definition above, $N_2 = N$, thus coalition $A_2$ controls every atom. We want to change control of $n_1$ so that $A_1$ controls $n_1$, changing the negotiation to a weakly deterministic one on the way. We add an additional agent $a$ that participates in $n_0, n_1, n_f$ as shown in Figure \ref{fig:control} in the middle and set $A_1 = \{A,a\}$. Now $A_1$ controls $n_1$ but also $n_0$ and $n_f$. We therefore add another agent $b$ that participates in $n_0,n_f$ as shown in Figure \ref{fig:control} on the right. Now the partition of atoms is exactly $N_1 = \{n_1\}$ and $N_2 = \{n_0,n_2,n_f\}$, as desired.
In general, by adding nondeterministic agents to the negotiation, we can change the control for each atom individually. For each atom $n$ whose control we wish to change, we add a number of agents to that atom, the initial atom $n_0$ and final atom $n_f$. We add nondeterministic edges for these agent from $n_0$ to $\{n, n_f\}$ for each outcome of $n_0$ and from  $n$ to $\{n, n_f\}$ for each outcome of $n$. It may be necessary to add more agents to $n_0$ that move to $n_f$ in order to preserve the control of $n_0$ or $n_f$. This procedure changes the control of $n$ while preserving soundness and weak determinism.

We proceed by showing that in the deterministic case, we cannot generate any atom control by adding more agents.

\begin{lemma}
We cannot add deterministic agents to the negotiation on the loft of Figure \ref{fig:control} in a manner, such that soundness is preserved and Player 1 controls $n_1$, Player 2 controls $n_2$.
\end{lemma}

\begin{proof}
Consider again the deterministic negotiation game on the left of Figure \ref{fig:control}. Assume we have added deterministic agents such that Player 1 controls $n_1$. After the occurring sequence $\vx_0 \by{(n_0,\text{a})} \vx_1 \by{(n_1,\text{a})} \vx_2$, those additional agents have moved deterministically, either to $n_1$ or $n_f$. \footnote{Moving to $n_2$ would change the control there, thus additional agents have to be added to $n_2$, we then can use a similar argument as follows by exchanging the roles of $n_1$ and $n_2$.} 

If any agent remains in $n_1$, choosing outcome $b$ in $n_2$ leads to a deadlock, otherwise, choosing $a$ leads to a deadlock. Thus the negotiation is no longer sound.
\end{proof}

\section{The Termination Problem}
\label{sec:complexity}

We turn to the general complexity of the problem. It is easy to see that
the termination problem can be solved in exponential time.

\begin{theorem}
\label{thm:membership}
The termination problem is in EXPTIME.
\end{theorem}
\begin{proof}
Sketch. (See the appendix for details.) We construct a concurrent reachability 
game on a graph such that Player 1 wins the negotiation game if{}f she wins this new game. 
The game has single exponential size in the size of the negotiation arena.  

The nodes of the graph are either
markings $\vx$ of the negotiation, or pairs $(\vx, N_{\vx})$, where $\vx$ is a marking
and $N_{\vx}$ is an independent set of atoms enabled at $\vx$. Nodes $\vx$ belong to Scheduler,
who chooses a set $N_{\vx}$, after which the play moves to $(\vx, N_{\vx})$. At nodes 
$(\vx, N_{\vx})$ Players 1 and 2 concurrently select outcomes for their atoms in $N_{\vx}$, and depending
on their choice the play moves to a new marking. Player 1 wins if the play reaches the 
final marking $\vx_f$. Since the winner of a concurrent game with reachability objectives
played on a graph can be determined in polynomial time (see e.g. \cite{DBLP:journals/jacm/AlurHK02}), 
the result follows.
\end{proof}

Unfortunately, there is a matching lower bound.

\begin{theorem}
\label{thm:hard}
The termination problem is EXPTIME-hard even for negotiations 
in which every reachable marking enables at most one atom.
\end{theorem}
\begin{proof}
Sketch. (See the appendix for details.)
The proof is by reduction from the acceptance problem 
for alternating, linearly-bounded Turing machines (TM)\cite{DBLP:journals/jacm/ChandraKS81}. Let 
$M$ be such a TM with transition relation $\delta$, and let $x$ be an input of length $n$. 

We define a negotiation with an agent $I$ modeling the internal state of $M$,
an agent $P$ modeling the position of the head, and one agent $C_i$ 
for each $1 \leq i \leq n$ modeling the $i$-th cell of the tape. 
The set of atoms contains one atom $n_{q,\alpha,k}$ for each triple 
$(q, \alpha, k)$, where $q$ is a state of $M$, $k$ is the current position of the 
head, and $\alpha$ is the current symbol in the $k$-th tape cell. 

The parties of the atom $n_{q, \alpha, k}$ are $I$, $P$ and $C_k$, in particular,
 $I$ and $P$ are agents of all atoms.
The atom $n_{q, \alpha, k}$ has one outcome 
$\r_\tau$ for each element $\tau \in \delta(q,\alpha)$, where $\tau$ 
is a triple consisting of a new state, a new tape symbol, and a direction 
for the head. We informally define the function $\trans(n_{q, \alpha, k},C_k, \r_\tau)$ 
by means of an example. Assume that, for 
instance, $\tau = (q', \beta, R)$, i.e., at control state $q$ and with the head reading 
$\alpha$, the machine can go to control state $q'$, write $\beta$, and move the head to 
the right. Then we have:  (i) $\trans(n_{q, \alpha, k},I,\r_\tau)$ contains all atoms of the form
$n_{q', \_, \_}$ (where $\_$ stands for a wildcard); (ii) $\trans(n_{q, \alpha, k},P,\r)$ 
contains the atoms $n_{\_, \_, k+1}$; and (iii) $\trans(n_{q, \alpha, k},C_k,\r_\tau)$ 
contains all atoms $n_{\_, \beta, \_}$. If atom $n_{q, \alpha, k}$ is 
the only one enabled and the outcome $\r_\tau$ occurs, then clearly in the new marking
the only atom enabled is $n_{q', \beta, k+1}$. So every reachable marking enables at most one atom.

Finally, the negotiation also has an initial atom that, loosely speaking, takes care of modeling the initial configuration.

The partition of the atoms is: an atom $n_{q, \alpha, k}$ belongs to $N_1$ if
$q$ is an existential state of $M$, and to $N_2$ if it is universal. It is easy to see that $M$ accepts $x$
if{}f Player 1 has a winning strategy.
\end{proof}

Notice that if no reachable marking enables two or more atoms, 
Scheduler never has any choice. Therefore, the termination problem is EXPTIME-hard
even if the strategy for Scheduler is fixed. 

A look at points (i)-(iii) in the proof sketch of this theorem shows that 
the negotiations obtained by the reduction are highly nondeterministic. 
In principle we could expect a lower complexity in the deterministic case.
However, this is not the case.

\begin{theorem}
\label{thm:hard2}
The termination problem is EXPTIME-hard even for deterministic negotiations 
in which every reachable marking enables at most one atom.
\end{theorem}
\begin{proof}
Sketch. (See the appendix for details.) We modify the 
construction of Theorem \ref{thm:hard} so that it yields 
a deterministic negotiation. The old construction has an atom
$n_{q, \alpha, k}$ for each state $q$, tape symbol $\alpha$, 
and cell index $k$, with $I$, $P$, and $C_k$ as parties. 

The new construction adds atoms $n_{q, k}$, with $I$ and $P$ as parties,
and $n_{\alpha, k}$, with $P$ and $C_k$ as parties. Atoms $n_{q, k}$ have an outcome 
for each tape symbol $\alpha$, and atoms $n_{\alpha, k}$ have an outcome for each state $q$. 
New atoms are controlled by Player 1. 

In the new construction, after the outcome of $n_{q, \alpha, k}$ for transition 
$(q',\beta,R) \in \delta(q,\alpha)$, agent $C_k$ moves to $n_{\beta, k}$, and agents 
$I$ and $P$ move to $n_{q', k+1}$. Intuitively, $C_k$ waits for the head to return to cell $k$, 
while agents $I$ and $P$ proceed. Atom $n_{q', k+1}$ has an outcome for every tape symbol $\gamma$. 
Intuitively, at this atom Player 1 guesses the current tape symbol in cell $k+1$;
the winning strategy corresponds to guessing right. After guessing, say, symbol $\gamma$,
agent $I$ moves directly to $n_{q',\gamma, k+1}$, while $P$ moves to $n_{\gamma, k+1}$.
Atom $n_{\gamma, k+1}$ has one outcome for every state of $M$. Intuitively, Player 1 now
guesses the current control state $q'$, after which both $P$ and  $C_{k+1}$ 
move to $n_{q',\gamma, k+1}$. Notice that all moves are now deterministic. 

If Player 1 follows the winning strategy, then the plays mimic those of the old construction:
a step like $(n_{q, \alpha, k}, (q',\beta,R))$ in the old construction, played when the current symbol
in cell $k+1$ is $\gamma$, is mimicked by a sequence 
$(n_{q, \alpha, k}, (q',\beta,R))$ $(n_{q',k+1}, \gamma)$ $(n_{\gamma, k+1}, q')$ of moves
in the new construction.

\end{proof}

\section{Termination in Sound Deterministic Negotiations}
\label{sec:soundDet}

In \cite{negI,negII} it was shown that the soundness problem 
(deciding whether a negotiation is sound) can be solved in polynomial time 
for deterministic negotiations and
acyclic, weakly deterministic negotiations (the case of cyclic weakly 
deterministic negotiations is open), while the problem is PSPACE-complete for
arbitrary negotiations. Apparently, Theorem \ref{thm:hard2} proves that 
the tractability of deter\-mi\-nis\-tic negotiations stops at game problems. 
We show that this is not the case. Well-designed negotiations are sound, 
since otherwise they contain atoms that can never occur (and can therefore be removed),
or a deadlock is reachable. Therefore, we are only interested
in the termination problem for sound negotiations. We prove that for sound deterministic negotiations
(in fact, even for the larger class of weakly deterministic type 2 negotiations) the 
termination and concluding-outcome problems are solvable in polynomial time. For this, we show that the 
well-known attractor construction for reachability games played on graphs as arenas can be ``lifted'' 
to sound and weakly deterministic type 2 negotiation arenas.

\begin{definition}
Let $\N$ be a negotiation arena with a set of atoms $N=N_1 \cup N_2$. Given $n \in N$, let $P_{n,det}$ 
be the set of deterministic agents participating in $n$.

The attractor of the final atom $n_f$ is $\mathcal{A} = \bigcup\limits_{k=0}^\infty \mathcal{A}_{k}$, where $\mathcal{A}_{0}  =  \{ n_f \}$ and 
\begin{equation*}
\begin{array}[t]{rl} 
\mathcal{A}_{k+1}  =  \mathcal{A}_{k} \;
 \cup & \{ n \in N_1 : \exists \r \in R_n \forall a \in P_{n,det} : \; \mathcal{X}(n,a,\r) \in \mathcal{A}_{k} \} \\
 \cup & \{ n \in N_2 : \forall \r \in R_n \forall a \in P_{n,det}: \; \mathcal{X}(n,a,\r) \in \mathcal{A}_{k}\}
\end{array}
\end{equation*}

Given a marking $\vx \neq \vx_f$ of $\N$ and a deterministic agent $a$, 
the attractor position of $a$ at $\vx$ is the smallest $k$ such that 
$\vx(a) \in \mathcal{A}_k$, or $\infty$ if $\vx(a) \notin \mathcal{A}$. 
Let $a_1, \ldots, a_k$ be the deterministic agents of $\N$.
The (attractor) position vector of $\vx$ is the tuple $(p_1, ..., p_k)$ where $p_i$ is the attractor position of $a_i$.
\end{definition}

\begin{theorem}
\label{thm:attract}
Let $\N$ be a sound, weakly deterministic type 2 negotiation arena. Player 1 has a winning strategy in the 
termination game iff $n_0 \in \mathcal{A}$.
\end{theorem}
\begin{proof}
We start with an observation: If all deterministic agents are ready to take part in $n_f$, 
then $n_f$ and only $n_f$ is enabled. 
Indeed, since deterministic agents are only ready to engage in at most one atom, the only atom
with a deterministic party that can be enabled is $n_f$. Moreover, by weak determinism type 2 every atom has 
a deterministic party, and so no atom other than $n_f$ can be enabled. Finally, by soundness at least
one atom is enabled, and so $n_f$ is the only enabled atom. 

\noindent ($\Leftarrow)$: Assume that $n_0 \in \mathcal{A}$. We fix the \emph{attractor strategy} for Player 1.
We define the attractor index of an atom $n \in \mathcal{A}$ as the smallest $k$ such that 
$n \in \mathcal{A}_k$. The strategy for an atom $n \in N_1 \cap \mathcal{A}$ 
is to choose any outcome such that all deterministic parties of $n$ move to an atom 
of smaller attractor index; formally, we choose any outcome $\r$ such that 
for every deterministic party $a$ the singular atom in $\trans(n,a,\r)$ has smaller attractor index than $n$. 
Such an outcome exists by construction of $\mathcal{A}$. For atoms $n \in N_1 \setminus \mathcal{A}$ 
we choose an arbitrary atom. Notice that this strategy is not only memoryless,
but also independent of the current marking. 

We show that the attractor strategy is winning. By the definition of the game, 
we have to prove that every play following the strategy ends with $n_f$ occurring.
By the observation above, it suffices to prove that the play reaches a marking at which
every deterministic agent is ready to engage in $n_f$. 

Assume there is a play $\pi$ where Player 1 plays according to the attractor strategy, which 
never reaches such a marking. Then the play never reaches the final marking $\vx_f$ either.
We claim that at all markings reached along $\pi$, the deterministic agents are only
ready to engage in atoms of $\mathcal{A}$. We first observe that, initially, all deterministic agents 
are ready to engage in $n_0$, and $n_0 \in \mathcal{A}$. Now,
assume that in some marking reached along $\pi$ the deterministic agents are only
ready to engage in atoms of $\mathcal{A}$. Then, by weak determinism type 2, all enabled atoms
belong to $\mathcal{A}$, and therefore also the atoms chosen by Scheduler.
By the definition of $\mathcal{A}$, after an atom of $N_2 \cap \mathcal{A}$ occurs, the
deterministic agents are ready to engage in atoms of $\mathcal{A}$ only; by the definition
of the attractor strategy, the same holds for atoms of $N_1 \cap \mathcal{A}$. This concludes 
the proof of the claim.

Since all markings $\vx$ reached along $\pi$ satisfy $\vx \neq \vx_f$ and 
$\vx(a) \in \mathcal{A}$ for every deterministic agent $a$, they all
have an associated attractor position vector whose components are natural numbers.
Let $P_k$ denote the position vector of the marking reached after $k \geq 0$ steps
in $\pi$. Initially only the initial atom $n_0$ is enabled, and so 
$P_0 = (k_0, k_0, \ldots, k_0)$, where $k_0$ is the attractor position of $n_0$. 
We have $k_0 < |N|$, the number of atoms of the negotiation arena $\N$. 
Given two position vectors $P= (p_1, ..., p_k)$ and $P'=(p_1', ..., p_k')$, we say $P \prec P'$ if
$p_i \leq p_i'$ for every $1 \leq i \leq k$, and $p_i < p_i'$ for at least one $1 \leq i \leq k$.
By the definition of the attractor strategy, the sequence $P_0, P_1, ...$ of attractor positions 
satisfies $P_{i+1} \prec P_i$ for every $i$. Since $\prec$ is a well-founded order, 
the sequence is finite, i.e., the game terminates. By the definition of the game, it terminates
at a marking that does not enable any atom. Since, by assumption, 
the play never reaches the final marking $\vx_f$, this marking is a deadlock, which contradicts the 
soundness of $\N$.

\noindent ($\Rightarrow$): As this part is dual to part $\Leftarrow$, we only sketch the idea. The complete proof can be found in the appendix.

Let $\mathcal{B} = N \setminus  \mathcal{A}$, and 
assume $n_0 \in \mathcal{B}$. 
A winning strategy for Player 2 is to choose for an atom $n \in N_2 \cap \mathcal{B}$ any outcome $r$ such that at least one deterministic agent
moves to an atom not in $\mathcal{A}$. Such an outcome exists by construction of $\mathcal{A}$. For atoms in $N_2 \setminus \mathcal{B}$
we chose an arbitrary outcome.

This strategy achieves the following invariant: If at some marking $\vx$ reached along a play according to this strategy, there is a deterministic agent $a$ that satisfies $\vx(a)=n$ and $n \in \mathcal{B}$, then the same holds after one more step of the play. We then conclude that since $n_0 \in \mathcal{B}$ this invariant holds in every step of the play. Therefore $n_f$, in which all deterministic agents participate and which is not in $\mathcal{B}$, can never be enabled. Finally, because of soundness a play cannot end in a deadlock and thus every play will be of infinite length.
\end{proof}

\begin{corollary}
\label{cor:poly}
For the termination game over sound and weakly deterministic type 2 negotiations, the following holds:
\begin{enumerate}[(a)]
\item The game collapses to a two-player game.
\item Memoryless strategies suffice for both players.
\item The winner and the winning strategy can be computed in $O(|\outc| \ast |\agents|)$ time, where 
$A$ is the set of agents and $|\outc|$ the total number of outcomes of the negotiation.
\end{enumerate}
\end{corollary}

\begin{proof}
(a) and (b): The attractor computation and the strategies used in the proof of Theorem \ref{thm:attract} are independent of the choices of Scheduler; the strategies are memoryless. \\
(c) An algorithm achieving this complexity can be found in the appendix.
\end{proof}

We still have to consider the possibility that requiring soundness alone, without the addition of
weak determinism type 2, already reduces the complexity of the termination problem. The following theorem shows that this is not the case, and concludes our study. The proof is again an adaption of the reduction from turing machines and can be found in the appendix.

\begin{theorem}
\label{thm:hard3}
The termination problem is EXPTIME-hard for sound negotiation arenas.
\end{theorem}

\section{The Concluding-Outcome Problem}

We demonstrate that the algorithm can also be used to solve the concluding-outcome problem. Remember that for each agent $a$ we have selected a set $G_a$ of outcomes that lead to the final atom and that we want occur. The key idea is to modify the negotiation in the following way: We add two atoms per agent, say \emph{good$_a$} and \emph{bad$_a$}, and redirect the outcomes of that agent that lead to the final atom to \emph{good$_a$} if they are in $G_a$ and to \emph{bad$_a$} otherwise. We illustrate this approach by example.

We slightly change the setting of Father, Mother and two Daughters such that both daughters want to go to the party. In the negotiation, that results in two additional atoms where D$_2$ talks alone with each parent. Figure \ref{extFDDM} shows the new negotiation omitting all edges of father and mother to simplify the representation. They both have a nondeterministic edge from each of their ports leading to all of their respective ports.

Imagine the goal 
of Daughter D1 is to get a ``yes'' answer for herself, but a `no'' answer for D2,
who always spoils the fun. Will a coalition with one parent suffice to achieve the goal?
To answer this question, we modify the negotiation before applying the construction as shown in 
Figure \ref{extFDDM}. We introduce dummy atoms \emph{good$_i$} and \emph{bad$_i$} for each Daughter $i$, and
we redirect ``yes'' transitions leading to $n_f$ to \emph{good$_1$} and \emph{bad$_2$} (these transitions are represented as dashed lines), analogously we redirect ``no'' transitions to \emph{bad$_1$} and \emph{good$_2$} (represented as dotted lines).

We apply the algorithm with a slight alteration: Instead of starting from $n_f$, we initially add all newly introduced good atoms to the attractor.
Applying the algorithm for the coalition Father-D1 
yields $\{$\emph{good$_1$},\emph{bad$_2$},$n_4\}$ as attractor; for Mother-D1
we get $\{$\emph{good$_1$},\emph{bad$_2$},$n_5\}$. So neither parent has enough influence to 
achieve the desired outcome.

\begin{figure}
\centering
\scalebox{0.8}{
\input{tikz/FDDM_extended.tex}
}
\caption{Applying the algorithm to the final-outcome problem. Father and mother edges omitted.}
\label{extFDDM}
\end{figure}

\section{Conclusions and Related Work}

We have started the study of games in the negotiation model 
introduced in \cite{negI,negII}. Our results confirm the low computational complexity 
of deterministic negotiations, however with an important twist: while even 
the simplest games are EXPTIME-hard for arbitrary deterministic negotiations, they become polynomial in the  \emph{sound} case. So soundness, a necessary feature of a well designed negotiation, also turns out to have a drastic beneficial effect on the complexity of the games.

We have shown that our games are also polynomial for sound and weakly deterministic negotiations. However, the complexity of deciding soundness for this case is unknown. We conjecture that it is also polynomial, as for the deterministic case. 

The objectives we have considered so far are qualitative: Either a coalition can reach termination or not, either it can enforce a concluding outcome or not. A possibility for future studies would be to look at quantitative objectives.

We have only considered 2-player games, since in our settings the behavior of the third player (the Scheduler) is either irrelevant, or is controlled by one of the other two players. We intend to study the extension to a proper 3-player game, or to a multi player game. Combined with qualitative objectives, this allows for multiple interesting questions: Which coalition of three agents can reach the best payoff? Which two agents should a given agent side with to maximize his payoff? Are there ``stable'' coalitions, that means, no agents can change to the opposing coalition and improve his payoff?

\bibliographystyle{eptcs}
\bibliography{references}

\appendix

\section{Appendix}

\subsection{Proofs of section \ref{sec:complexity}}

\begin{reftheorem}{thm:membership}
The termination problem is in EXPTIME.
\end{reftheorem}

For the proof we need some definitions and results from \cite{DBLP:journals/jacm/AlurHK02}.

A concurrent game structure is a tuple $S=(k,Q,\Pi, \pi, d, \delta)$ where
\begin{itemize}
\item $k \geq 1$ is a natural number, the number of players.
\item $Q$ is a finite set of states.
\item $\Pi$ is a finite set of propositions.
\item $\pi$ assigns every state $q\in Q$ a set of propositions that are true in $q$.
\item $d$ assigns every player $a \in \{1, ..., k\}$ and every state $q \in Q$ a natural number $d_a(q) \geq 1$ of possible moves. We identify these moves with natural numbers $1, ..., d_a(q)$. For each state $q$, we write $D(q)$ for the set $\{1, ..., d_1(q)\} \times ... \times \{1, ..., d_k(q)\}$ of move vectors.
\item $\delta$ is the transition function that assigns a state $q \in Q$ and a move vector $(j_1, ..., j_k) \in D(q)$ a state $\delta(q, j_1, ..., j_k)$ that results from state $q$ if every player $a \in \{1,...,k\}$ chooses move $j_a$.
\end{itemize}

The following is one of the results from \cite{DBLP:journals/jacm/AlurHK02}:

\begin{theorem}
Model checking for reachability objectives on concurrent game structures is possible in time linear in the number of transitions of the game and the length of the reachability formula.
\label{thm:modelpoly}
\end{theorem}

Using this result, it is easy to see that the termination problem can be solved in exponential time.

\begin{proof}
We construct a 3-player concurrent game structure and a reachability objective such that Player 1 wins the negotiation game if{}f he wins the newly constructed game
against the coalition of the other two players. 
The game has single exponential size in the size of the negotiation arena.

The states of the game are either
markings $\vx$ of the negotiation, or pairs $(\vx, N_{\vx})$, where $\vx$ is a marking
and $N_{\vx}$ is an independent set of atoms enabled at $\vx$. Nodes $\vx$ belong to Scheduler,
who chooses a set $N_{\vx}$, after which the play moves to $(\vx, N_{\vx})$. At nodes 
$(\vx, N_{\vx})$ Players 1 and 2 concurrently select subsets of $N_{\vx}$, and depending
on their choice the play moves to a new marking.

The states of the concurrent game structure are the following:
\[\begin{array}{rl}
Q = & \{\vx : \vx \text{ a reachable marking}\} \cup \{(\vx, N_{\vx}) : \vx \text{ a reachable marking}, \\
 & N_{\vx} \text{ a set of independent atoms enabled in } \vx\} \\
\end{array}\]
Only one proposition is needed to mark the final state:
\[ \Pi = \{\text{final}\}\]
\[ \pi(\vx) = \begin{cases}
\text{final} & \text{if }\vx = \vx_f \\
\emptyset & \text{otherwise}
\end{cases} \]
\[ \pi((\vx, N_{\vx})) = \emptyset \]
The move vectors, given by the cartesian product of the possible moves of each player,
are defined as follows: 
\begin{itemize}
\item $D(\vx) = \{1\}\times\{1\}\times S$ is the set of move vectors from state $\vx$, where 
$S$ is the set of all sets of independent atoms enabled in $\vx$.
\item $D((\vx, N_{\vx})) = \{F_1\}\times\{F_2\}\times\{1\}$ is the set of  move vectors from state 
$(\vx, N_{\vx})$, where $F_i$ is the set of all functions assigning each $n\in N_{\vx} \cap N_i$ 
an outcome $\r \in R_n$. 
\end{itemize}
The transition function $\delta$ is defined by $\delta(\vx,(1,1,s)) = (\vx,s)$, and \\
$\delta((\vx, N_{\vx}),(f_1,f_2,1)) = \vx'$,  where $\vx'$ is the result starting from $\vx$ after all the atoms of $N_{\vx}$ occur with the outcomes specified by $f_1,f_2$. Notice that,
since the atoms are independent, $\vx'$ does not depend on the order in which they occur.

Player 1 wins if the play reaches the 
final marking $\vx_f$, or equivalently, the reachability objective is $\Diamond\text{final}$. Using theorem \ref{thm:modelpoly}, the result follows.
\end{proof}

\begin{reftheorem}{thm:hard}
The termination problem is EXPTIME-hard even for negotiations 
in which every reachable marking enables at most one atom.
\end{reftheorem}

\begin{proof}
The proof is by reduction from the acceptance problem 
for alternating, linearly-bounded Turing machines (TM) \cite{DBLP:journals/jacm/ChandraKS81}. We are given an alternating TM 
$M$ with transition relation $\delta$, and an input $x$ of length $n$. We assume that 
$M$ always halts in one of two designated states $q_{accept}$ or $q_{reject}$, 
and does so immediately after reaching one of those states.

We define a negotiation with two agents $P$ and $I$ modeling the head position 
and internal state of $M$, and one agent 
$C_k$ for each cell. The set of atoms contains one atom $n_{q,\alpha,k}$ for each triple 
$(q, \alpha, k)$, where $q$ is a control state of $M$, $k$ is the current position of the 
head, and $\alpha$ is the current letter in the $k$-th tape cell. 

\begin{figure}[ht]
\centerline{
\tikz{
\nego[id=n_0ak,text={$C_k$}{$I$}{$P$},spacing=1]{0,0}
\nego[id=n_1ak,text={$C_k$}{$I$}{$P$},spacing=1]{3,0}
\nego[id=n_0bk,text={$C_k$}{$I$}{$P$},spacing=1]{0,-1.5}
\nego[id=n_1bk,text={$C_k$}{$I$}{$P$},spacing=1]{3,-1.5}
\nego[id=n_0ck,text={$C_k$}{$I$}{$P$},spacing=1]{0,-3}
\nego[id=n_1ck,text={$C_k$}{$I$}{$P$},spacing=1]{3,-3}
\nego[id=n_0ak1,text={$C_{k+1}$}{$I$}{$P$},spacing=1]{7,0}
\nego[id=n_1ak1,text={$C_{k+1}$}{$I$}{$P$},spacing=1]{10.1,0}
\nego[id=n_0bk1,text={$C_{k+1}$}{$I$}{$P$},spacing=1]{7,-1.5}
\nego[id=n_1bk1,text={$C_{k+1}$}{$I$}{$P$},spacing=1]{10.1,-1.5}
\nego[id=n_0ck1,text={$C_{k+1}$}{$I$}{$P$},spacing=1]{7,-3}
\nego[id=n_1ck1,text={$C_{k+1}$}{$I$}{$P$},spacing=1]{10.1,-3}
\node[draw,dashed,fit={($(n_0ak_P0.north west)+(-0.4,0.4)$) ($(n_1ck_P2.south east)+(0.4,-0.4)$)}] (box1) {};
\node[anchor=north] at (box1.south) {tape cell k (current cell)};
\node[draw,dashed,fit={($(n_0ak1_P1.north west)+(-1.5,0.45)$) ($(n_1ck1_P2.south east)+(0.4,-0.4)$)}] (box2) {};
\node[anchor=north] at (box2.south) {tape cell k+1};
\node[anchor=east] at ($(n_0ak_P0.west) + (-1,0)$) {a};
\node[anchor=east] at ($(n_0bk_P0.west) + (-1,0)$) {b};
\node[anchor=east] at ($(n_0ck_P0.west) + (-1,0)$) {c};
\node[anchor=south] at ($(n_0ak_P1.north) + (0,0.5)$) {$q_1$};
\node[anchor=south] at ($(n_1ak_P1.north) + (0,0.5)$) {$q_2$};
\node[anchor=south] at ($(n_0ak1_P1.north) + (0,0.5)$) {$q_1$};
\node[anchor=south] at ($(n_1ak1_P1.north) + (0,0.5)$) {$q_2$};
\draw[dashed] (box1.north west) -- +(-1,1);
\node[anchor=south] at ($(box1.north west) + (0.2,0)$) {\begin{tabular}{c} TM \\ state \end{tabular}};
\node[anchor=east] at (box1.north west) {\begin{tabular}{c} cell \\ content \end{tabular}};
\pgfsetarrowsend{latex}
\draw (n_1ak_P0) -- (n_1bk_P0);
\draw (n_1ak_P0) +(0,-0.75) -- +(-3,-0.75) -- (n_0bk_P0);
\draw (n_1ak_P1) -- +(0,-0.75) -- +(4,-0.75) -- (n_0ak1_P1);
\draw (n_1ak_P1) -- +(0,-0.75) -- +(4,-0.75) -- (n_0bk1_P1);
\draw (n_1ak_P1) -- +(0,-0.75) -- +(2,-0.75) -- +(2,-2.25) -- +(4,-2.25) -- (n_0ck1_P1);
\draw (n_1ak_P2) -- +(0,+0.5) -- +(4.5,+0.5) -- +(4.5,0) -- (n_0ak1_P2);
\draw (n_1ak_P2) -- +(0,+0.5) -- +(4.5,+0.5) -- +(4.5,-1.5) -- (n_0bk1_P2);
\draw (n_1ak_P2) -- +(0,+0.5) -- +(4.5,+0.5) -- +(4.5,-3) -- (n_0ck1_P2);
}
}
\caption{Part of the negotiation representing a TM, the edges drawn represent the transition $(q_2,a) \ra (q_1,b,R)$}
\label{fig:TM}
\end{figure}

Figure \ref{fig:TM} shows a part of the negotiation corresponding to a TM with two internal states, $q_1,q_2$ and three tape symbols $a,b,c$. The part shown represents two cells of the tape as well as the movement of the agents corresponding to the TM taking the transition $(q_2,a) \ra (q_1,b,R)$.

The parties of the atom $n_{q, \alpha, k}$ are $I$, $P$ and $C_k$, in particular $I$ and $P$ are agents of all atoms.
The atom $n_{q, \alpha, k}$ has one outcome 
$\r_\tau$ for each element $\tau \in \delta(q,\alpha)$, where $\tau$ 
is a triple consisting of a new state, a new tape symbol, and a direction 
for the head. We informally define the function $\trans(n_{q, \alpha, k},C_k, \r_\tau)$ 
by means of an example. Assume that, for 
instance, $\tau = (q', \beta, R)$, i.e., at control state $q$ and with the head reading 
$\alpha$, the machine can go to control state $q'$, write $\beta$, and move the head to 
the right. Then we have:  (i) $\trans(n_{q, \alpha, k},I,\r_\tau)$ contains all atoms of the form
$n_{q', \_, \_}$ (where $\_$ stands for a wildcard); (ii) $\trans(n_{q, \alpha, k},P,\r)$ 
contains the atoms $n_{\_, \_, k+1}$; and (iii) $\trans(n_{q, \alpha, k},C_k,\r_\tau)$ 
contains all atoms $n_{\_, \beta, \_}$. If atom $n_{q, \alpha, k}$ is 
the only one enabled and the outcome $\r_\tau$ occurs, then clearly in the new marking
the only atom enabled is $n_{q', \beta, k+1}$. So every reachable marking enables at most one atom.

Finally, the negotiation also has an initial atom that, loosely speaking, takes care 
of modeling the initial configuration.

The partition of the atoms is: an atom $n_{q, \alpha, k}$ belongs to $N_1$ if
$q$ is an existential state of $M$, and to $N_2$ if it is universal. It is easy to see that $M$ accepts $x$
if{}f Player 1 has a winning strategy.

Notice that if no reachable marking enables two or more atoms, 
Scheduler never has any choice. Therefore, the termination problem is EXPTIME-hard
even if the strategy for Scheduler is fixed. 

Formally, we are given an alternating TM $M=(Q, \Gamma, \delta, q_0, g)$ with tape alphabet $\Gamma$ and a function $g \colon Q \ra \{\wedge, \vee, \text{accept}, \text{reject}\}$ that 
determines the state type, and an input $x$ of length $n$. We assume that $M$ always halts, 
does so in designated states $q_{accept}$ or $q_{reject}$, and does so immediately when reaching 
one of those. We define a negotiation game $\N_{M,x}$ as follows:

\begin{itemize}
\item The set of agents is
\[A = \{I,P\} \cup \{C_k : k \in \{0,...,n-1\}\}\]
\item The set of atoms is
\[N = \{n_0, n_f\} \cup \{n_{q,\alpha, k} : (q,\alpha, k) \in Q\times\Gamma\times\{0,...,n-1\}\}\]
\item The set $N_1$ of atoms owned by Player 1 contains $n_0$, $n_f$, and all atoms $n_{q,\alpha, k}$ such that $g(q) \not= \wedge$.
The set $N_2$ contains all other atoms.
\item The set of parties of each atom is given by $P_{n_0} = P_{n_f} = A$, and \\ 
$P_{n_{q,\alpha,k}} = \{I,P, C_k\}$ for every $n_{q,\alpha, k} \in N$.
\item The set of outcomes of each atom is defined as follows: $R_{n_0} = \{\text{st}\}$,\\
$R_{n_f} = \{\text{end}\}$, and for every $n_{q,\alpha, k} \in N$
\[R_{n_{q,\alpha,k}} = 
\begin{cases}
 \{\text{f}\} & \mbox{ if } q = q_{accept} \\
 \{\text{s}\} &  \mbox{ if } q = q_{reject}  \\
 \{\r_\tau : \tau \in \delta(q, \alpha)\} & \mbox{otherwise}
\end{cases}\]
\item The $\trans(\_, I, \_)$ function for agent $I$ is given by:
\[\begin{array}{rcl}
 \trans(n_0,I,\text{st}) & = & \{n_{q_0,\alpha, 0} : \alpha \in \Gamma\} \\
 \trans(n_{q,\alpha,k},I,r_\tau) & = & 
 \begin{cases}
  \{n_{q',\alpha, k-1} : \alpha \in \Gamma\} & \text{if }\tau = (q',\_, L)\\
  \{n_{q',\alpha, k+1} : \alpha \in \Gamma\} & \text{if }\tau = (q',\_, R)
 \end{cases} \\
 \trans(n_{q_{accept},\alpha,k},I,\text{f}) & = & \{n_f\}\\
 \trans(n_{q_{reject},\alpha,k},I,\text{s}) & = & \{n_{q_{reject},\alpha,k}\}
\end{array}\]
\item The $\trans(\_, P, \_)$ function for agent $P$ is given by
$\trans(n,P,r) = \trans(n,I,r)$ for all $n \in N, r \in P_n$.
\item Finally, the $\trans(\_, C_k, \_)$ function for agent $C_k$ (the only part of the definition that depends on the input $x$ to $M$) is given by
\[\begin{array}{rcl}
 \trans(n_0,C_k,\text{st}) & = & \{n_{q,\alpha, k} : q \in Q, \alpha\text{ the initial value of cell $k$}\} \cup \{n_f\} \\
 \trans(n_{q,\alpha,k},C_k, r_\tau) & = & \{n_f\} \cup \{n_{q',\beta, k} : q' \in Q\}\text{ if }\tau = (\_,\beta,\_) \\
 \trans(n_{q_{accept},\alpha,k},C_k,\text{f}) & = & \{n_f\}\\
 \trans(n_{q_{reject},\alpha,k},C_k,\text{s}) & = & \{n_{q_{reject},\alpha,k}\}
\end{array}\]
\end{itemize}

Observe that, in general, the negotiation $\N_{M,x}$ is nondeterministic and unsound. 
Initially, we have $\vx_0(a) = \{n_0\}$ for every agent $a$, and so only the atom $n_0$
is enabled. Thereafter, at any point in time, exactly one atom $n_{q, \alpha, k}$ is 
enabled until $n_{q_{accept}, \alpha, k}$ or $n_{q_{reject}, \alpha, k}$ is reached. 
Finally, $n_f$ is enabled if{}f $n_{q_{accept}, \alpha, k}$ was reached. Since 
the outcome of $n_{q, \alpha, k}$ is chosen by Player 1 when $q$ is an existential state,
and by Player 2 when it is a universal state, Player 1 has a winning strategy if{}f 
$M$ accepts $x$.
\end{proof}

\begin{reftheorem}{thm:hard2}
The termination problem is EXPTIME-hard even for deterministic negotiations 
in which every reachable marking enables at most one atom.
\end{reftheorem}

\begin{proof}
We now modify the initial construction to yield a deterministic (but still unsound) negotiation. First we present a general construction to remove nondeterminism from arbitrary negotiations, we then apply this construction with some alterations to fit the needs of negotiation games.
The transformation is as follows: for any triple $(n,a,\r)$ such that
 $\trans(n,a,\r) = \{n_1, \ldots, n_k\}$ for some
$k > 1$, we introduce a new atom $n_{a,\r}$, assigned to Player 1, with $a$ being the only agent that participates in $n_{a,\r}$, and with $k$ different outcomes $\r_1, \ldots, \r_k$, and redefine 
$\trans$ as follows: $\trans(n,a,\r) = \{n_{a,\r} \}$, and $\trans(n_{a,\r}, \r_i\} = \{n_i\}$ 
for every $1 \leq i \leq k$.
Intuitively, in the original negotiation after the outcome $(n_0, \texttt{st})$ agent 
$a$ is ready to engage in $n_1, \ldots, n_k$; in the new negotiation, after $(n_0, \texttt{st})$ $a$ commits to \emph{exactly one} of $n_1, \ldots, n_k$, and the choice is controlled by Player 1.
Unfortunately, we cannot apply this construction directly because of the following problem: An agent $C_k$ has to decide in which of the $n_{q, \alpha, k}$ he should to participate next, therefore he needs to know the state $q$ in which the head will arrive next in cell $k$. But this may not be decided until the head arrives in cell $k$. Therefore we need to postpone this decision until the head actually arrives in cell $k$.

This is achieved by the following alterations:
In the new construction, after firing the outcome of $n_{q, \alpha, k}$ for transition 
$(q',\beta,R) \in \delta(q,\alpha)$, agent $C_k$ moves to $n_{\beta, k}$, and agents 
$I$ and $P$ move to $n_{q', k+1}$. Intuitively, $C_k$ waits for the head to return to cell $k$, 
while agents $I$ and $P$ proceed. Atom $n_{q', k+1}$ has an outcome for every tape symbol $\gamma$. 
Intuitively, at this atom Player 1 guesses the current tape symbol in cell $k+1$;
the winning strategy corresponds to guessing right. After guessing, say, symbol $\gamma$,
agent $I$ moves directly to $n_{q',\gamma, k+1}$, while $P$ moves to $n_{\gamma, k+1}$.
Atom $n_{\gamma, k+1}$ has one outcome for every state of $M$. Intuitively, Player 1 now
guesses the current control state $q'$, after which both $P$ and  $C_{k+1}$ 
move to $n_{q',\gamma, k+1}$. Notice that all moves are now deterministic. 

If Player 1 follows the winning strategy, then the plays mimic those of the old construction:
a step like $(n_{q, \alpha, k}, (q',\beta,R))$ in the old construction, played when the current symbol
in cell $k+1$ is $\gamma$, is mimicked by a sequence 
$(n_{q, \alpha, k}, (q',\beta,R))$ $(n_{q',k+1}, \gamma)$ $(n_{\gamma, k+1}, q')$ of moves
in the new construction.

Formally, the construction is the following:

\begin{itemize}
\item The set of agents is
\[A = \{I,P\} \cup \{C_k : k \in \{0,...,n-1\}\}\]
\item The set of atoms is
\[\begin{array}{rl}
N = \{n_0, n_f\} & \cup \{n_{q,\alpha, k} : (q,\alpha, k) \in Q\times\Gamma\times\{0,...,n-1\}\} \\
 & \cup \{n_{\alpha, k} : (\alpha, k) \in \Gamma\times\{0,...,n-1\}\} \\
 & \cup \{n_{q,k} : (q,k) \in Q \times\{0,...,n-1\}\}
\end{array}\] 
\item The set $N_1$ of atoms owned by Player 1 contains $n_0$, $n_f$, all atoms $n_{\alpha, k}$, all atoms $n_{q,k}$, and all atoms $n_{q,\alpha, k}$ such that $g(q) \not= \wedge$.
The set $N_2$ contains all other atoms.
\item The set of parties of each atom is given by 
\[P_{n_0} = P_{n_f} = A\]
\[P_{n_{q,\alpha,k}} = \{I,P, C_k\} \text{ for every } n_{q,\alpha, k} \in N\]
\[P_{n_{\alpha,k}} = \{P, C_k\} \text{ for every } n_{\alpha, k} \in N\]
\[P_{n_{q,k}} = \{I,P\} \text{ for every } n_{q, k} \in N\]
\item The set of outcomes of each atom is defined as follows: $R_{n_0} = \{\text{st}\}$, \\
$R_{n_f} = \{\text{end}\}$, and
\[R_{n_{q,\alpha,k}} = \begin{cases}
\{\text{f}\} & q = q_{accept} \\
\{\text{s}\} & q = q_{reject}  \\
\{\r_\tau : \tau \in \delta(q, \alpha)\} & otherwise
\end{cases} \qquad\text{ for every } n_{q,\alpha, k} \in N\] 
\[R_{n_{q,k}} = \{\r_\alpha : \alpha \in \Gamma\} \text{ for every } n_{q,k} \in N\]
\[R_{n_{\alpha,k}} = \{\r_q : q \in Q\} \text{ for every } n_{\alpha, k} \in N\]
\item The $\trans(\_, P, \_)$ function for agent $P$ is given by:
\[\trans(n_0,P,\text{st}) = \{n_{\alpha, 0} : \alpha\text{ is the initial content of cell 0} \}\] 
\[\trans(n_{q,\alpha,k},P,\r_\tau) = \begin{cases}
\{n_{q', k-1} \} & \text{if }\tau = (q',\_, L)\\
\{n_{q', k+1} \} & \text{if }\tau = (q',\_, R)
\end{cases}\]
\[\trans(n_{q, k},P,\r_\alpha) = \{n_{\alpha, k}\}\]
\[\trans(n_{\alpha, k},P,\r_q) = \{n_{q,\alpha, k}\}\]
\[\trans(n_{q_{accept},\alpha,k},P,\text{f}) = \{n_f\}\]
\[\trans(n_{q_{reject},\alpha,k},P,\text{s}) = \{n_{q_{reject},\alpha,k}\}\]
\item The $\trans(\_, I, \_)$ function for agent $I$ is given by:
\[\trans(n_0,I,\text{st}) = \{n_{q_0,\alpha, 0} : \alpha\text{ is the initial content of cell 0} \}\]
\[\trans(n_{q,\alpha,k},I,\r_\tau) = \begin{cases}
\{n_{q', k-1} \} & \text{if }\tau = (q',\_, L)\\
\{n_{q', k+1} \} & \text{if }\tau = (q',\_, R)
\end{cases}\]
\[\trans(n_{q,k},I,\r_\alpha) = \{n_{q,\alpha, k}\}\]
\[\trans(n_{q_{accept},\alpha,k},I,\text{f}) = \{n_f\}\]
\[\trans(n_{q_{reject},\alpha,k},I,\text{s}) = \{n_{q_{reject},\alpha,k}\}\]
%
\item Finally, the $\trans(\_, C_k, \_)$ function for agent $C_k$ is given by
\[\trans(n_0,C_k,\text{st}) = \{n_{\alpha, k} : \alpha\text{ is the initial content of cell k} \}\]
\[\trans(n_{q,\alpha,k},C_k,\r_\tau) = \{n_{\beta, k} \} \text{ if }\tau = (\_,\beta, \_)\]
\[\trans(n_{\alpha, k},C_k,\r_q) = \{n_{q,\alpha, k}\}\]
\[\trans(n_{q_{accept},\alpha,k},C_k,\text{f}) = \{n_f\}\] 
\[\trans(n_{q_{reject},\alpha,k},C_k,\text{s}) = \{n_{q_{reject},\alpha,k}\}\]
\end{itemize}
\end{proof}

\subsection{Proofs of section \ref{sec:soundDet}}

In the paper we have shortened the proof of the only-if part of Theorem \ref{thm:attract}
due to space constraints. Here is the full text.

\begin{reftheorem}{thm:attract}
Let $\N$ be a sound, weakly deterministic negotiation arena. Player 1 has a winning strategy in the 
termination game iff $n_0 \in \mathcal{A}$.
\end{reftheorem}

\begin{proof}
We start with an observation: If all deterministic agents are ready to take part in $n_f$, 
then $n_f$ and only $n_f$ is enabled. 
Indeed, since deterministic agent are only ready to engage in at most one atom, the only atom
with a deterministic party that can be enabled is $n_f$. Moreover, by weak determinism every atom has 
a deterministic party, and so no atom other than $n_f$ can be enabled. Finally, by soundness at least
one atom is enabled, and so $n_f$ is the only enabled atom. 

\noindent ($\Leftarrow)$: Assume that $n_0 \in \mathcal{A}$. We fix the \emph{attractor strategy} for Player 1.
We define the attractor index of an atom $n \in \mathcal{A}$ as the smallest $k$ such that 
$n \in \mathcal{A}_k$. The strategy for an atom $n \in N_1 \cap \mathcal{A}$ 
is to choose any outcome such that all deterministic parties of $n$ move to an atom 
of smaller attractor index; formally, we choose any outcome $r$ such that 
for every deterministic party $a$ the singular atom in $\trans(n,a,r)$ has smaller attractor index than $n$. 
Such an outcome exists by construction of $\mathcal{A}$. For atoms $n \in N_1 \setminus \mathcal{A}$ 
we choose an arbitrary atom. Notice that this strategy is not only memoryless,
but also independent of the current marking. 

We show that the attractor strategy is winning. By the definition of the game, 
we have to prove that every play following the strategy ends with $n_f$ occurring.
By the observation above, it suffices to prove that the play reaches a marking at which
every deterministic agent is ready to engage in $n_f$. 

Assume there is a play $\pi$ where Player 1 plays according to the attractor strategy, which 
never reaches such a marking. Then the play never reaches the final marking $\vx_f$ either.
We claim that at all markings reached along $\pi$, the deterministic agents are only
ready to engage in atoms of $\mathcal{A}$. We first observe that, initially, all deterministic agents 
are ready to engage in $n_0$, and $n_0 \in \mathcal{A}$. Now,
assume that in some marking reached along $\pi$ the deterministic agents are only
ready to engage in atoms of $\mathcal{A}$. Then, by weak determinism, all enabled atoms
belong to $\mathcal{A}$, and therefore also the atoms chosen by Scheduler.
By the definition of $\mathcal{A}$, after firing atoms of $N_2 \cap \mathcal{A}$ the
deterministic agents are ready to engage in atoms of $\mathcal{A}$ only; by the definition
of the attractor strategy, the same holds atoms of $N_1 \cap \mathcal{A}$. This concludes 
the proof of the claim.

Since all markings $\vx$ reached along $\pi$ satisfy $\vx \neq \vx_f$ and 
$\vx(a) \in \mathcal{A}$ for every deterministic agent $a$, they all
have an associated attractor position vector whose components are natural numbers.
Let $P_k$ denote the position vector of the marking reached after $k \geq 0$ steps
in $\pi$. Initially only the initial atom $n_0$ is enabled, and so 
$P_0 = (k_0, k_0, \ldots, k_0)$, where $k_0$ is the attractor position of $n_0$. 
We have $k_0 < |N|$, the number of atoms of the negotiation arena $\N$. 
Given two position vectors $P= (p_1, ..., p_k)$, $P'=(p_1', ..., p_k')$, we say $P \prec P'$ if
$p_i \leq p_i'$ for every $1 \leq i \leq k$, and $p_i < p_i'$ for at least one $1 \leq i \leq k$.
By the definition of the attractor strategy, the sequence $P_0, P_1, ...$ of attractor positions 
satisfies $P_0 \succ P_1 \succ P_2 \ldots$. Since $\prec$ is a well-founded order, 
the sequence is finite, i.e., the game terminates. By the definition of the game, it terminates
at a marking that does not enable any atom. Since, by assumption, 
the play never reaches the final marking $\vx_f$, this marking is a deadlock, which contradicts the 
soundness of $\N$.

\noindent ($\Rightarrow$): Let $\mathcal{B} = N \setminus  \mathcal{A}$, and 
assume $n_0 \in \mathcal{B}$. 
We give a winning strategy for Player 2. The strategy for an atom $n \in N_2 \cap \mathcal{B}$ 
is to choose any outcome $\r$ such that at least one deterministic agent
moves to an atom not in $\mathcal{A}$, i.e., such that $\trans(n,a,\r) \notin \mathcal{A}$ for at least
one deterministic agent $a$. Such an outcome exists because, since $n \in \mathcal{B}$, 
we have $n \notin \mathcal{A}_k$ for every $k$, and so by definition and monotonicity of $\mathcal{A}_k$ there exists $\r \in R_n$ such that 
$\mathcal{X}(n,a,r) \notin \mathcal{A}_k$
for some deterministic agent $a$ and every $k$. For atoms in $N_2 \setminus \mathcal{B}$
we chose an arbitrary outcome.

We show that this strategy is winning for Player 2. Once again, because the negotiation is sound, 
no play played according to this strategy ends in a deadlock. So we have to prove that
every game played following the strategy never ends. By the observation above, it suffices to 
prove that the play never reaches a marking at which every deterministic agent is ready to engage in $n_f$.
Further, since $n_f \in \mathcal{A}$, it suffices to show  that for every marking $\vx$ reached along the 
play there is a deterministic agent $a$ such that $\vx(a) \in \mathcal{B}$.

Initially all deterministic agents are only ready to engage in $n_0$, and $n_0 \in \mathcal{B}$.
Now, assume that at some marking $\vx$ reached along $\pi$ there is a
deterministic agent $a$ that satisfies $\vx(a)=n$ and $n \in \mathcal{B}$. We prove that the same holds
for the marking $\vx'$ reached after one step of the play. If $n$ is not enabled
at $\vx$, then we have $\vx'(a) = \vx(a) \in \mathcal{B}$,
and we are done. The same holds if $n$ is enabled at $\vx$, but is not selected by the Scheduler. If
$n$ is enabled and selected by the Scheduler, there are two possible cases. If $n \in N_1$ 
then by the definition of $\mathcal{A}$ for every outcome $\r$ of $n$ there is a deterministic agent $a$ such that 
$\mathcal{X}(n_0,a,\r) \in \mathcal{B}$, and so $\vx'(a) \in \mathcal{B}$. If $n \in N_2$,
then by definition the strategy chooses an outcome $r$ for which there is an agent $a$ such that 
$\mathcal{X}(n_0,a,\r) \in \mathcal{B}$, and we again get $\vx'(a) \in \mathcal{B}$.
\end{proof}

\begin{refcorollary}{cor:poly}
For the termination game over sound and weakly deterministic negotiations, the following holds:
\begin{enumerate}[(a)]
\item The game collapses to a two-player game.
\item Memoryless strategies suffice for both players.
\item The winner and the winning strategy can be computed in $O(|\outc| \ast |\agents|)$ time, where 
$A$ is the set of agents and $|\outc|$ the total number of outcomes of the negotiation.
\end{enumerate}
\end{refcorollary}

\begin{proof}
(a) The attractor computation and the strategies used in the proof of Theorem \ref{thm:attract} are unrelated to the actions of the Scheduler, his choices have no influence on the outcome. \\
(b) Both strategies used in the proof of Theorem \ref{thm:attract} are memoryless. \\
(c) We describe an algorithm that computes the set $\mathcal{A}$ and the winning strategy. 
\lstset{language=C++,tabsize=4,escapeinside={(*@}{@*)},morekeywords={array,set,choose,remove,every,from,where,some}}
\begin{lstlisting}[mathescape]
array[][] count
array[] outcomes		
array[] strategy
set border
set $\mathcal{A}$
for every atom $n$
	outcomes[$n$] $\leftarrow$ |$\outc_n$|
	for every outcome $\r\in\outc_n$
		count[$n$][$\r$] $\leftarrow$ |$P_n$|
$\mathcal{A} \leftarrow \{n_f\}$
border $\leftarrow \{n_f\}$
outcomes[$n_f$] $\leftarrow 0$
for every $\r \in \outc_{n_f}$
	count[$n_f$][$\r$] $\leftarrow 0$
while border $\neq \emptyset$
	choose and remove an atom $n$ from border
	for every atom $n' \in \N$ where $\trans(n',a,\r) = n$ for some $a\in P_n', \r\in\outc_n'$(*@\label{forLoop}@*) 
		count[$n'$][$\r$]$--$
		if count[$n'$][$\r$] = 0 and $n'\in \N_1$
			outcomes[$n'$] $\leftarrow 0$
			strategy[$n'$] $\leftarrow \r$
			border $\leftarrow$ border $\cup$ $\{n'\}$
		if count[$n'$][$\r$] = 0 and $n'\in \N_2$
			outcomes[$n'$]$--$
			if outcomes[$n'$] = 0
				border $\leftarrow$ border $\cup$ $\{n'\}$
return $\mathcal{A}$
\end{lstlisting}
Intuitively, $count[n][\r]$ counts the ports of atom $n$ that do not reach the attractor with outcome $\r$, $outcomes[n]$ counts the outcomes of $n$ for which $count[n][\r] \not= 0$ or is zero if $n$ is in the attractor.

This algorithm terminates and computes $\mathcal{A}$, the check $n_0 \in\mathcal{A}$ is done by checking whether $outcomes[n_0]$ is zero. Furthermore, the strategy for Player 1 is computed in $strategy[n]$. Every outcome is inspected at most once per agent, thus the running time is at most $O(|R| \ast |A|)$. This leads to a total running time of at most $O(|R| \ast |A| + |N|) \in O(|R| \ast |A|)$. Computing the strategy for Player 2 can be done in a straightforward loop over all outcomes choosing any outcome for which $count[n][\r] \not= 0$ in running time $O(R)$.
\end{proof}

\begin{reftheorem}{thm:hard3}
The termination problem is EXPTIME-hard for sound negotiation arenas.
\end{reftheorem}

\begin{proof}
Recall that a negotiation is sound if every atom occurs in at least one occurrence sequence,
and every occurrence sequence can be extended to a sequence that ends with the final atom $n_f$ occurring.
There are two reasons why the negotiation $\N_{M,x}$ defined in  the proof of Theorem \ref{thm:hard} 
is unsound in general:

\begin{itemize}
\item If $M$ does not accept, then the final atom $n_f$ is not reachable.
\item If the head can never reach cell $k$ with symbol $\alpha$ in it, then the atom
$n_{q,\alpha,k}$ never occurs in any occurrence sequence.
\end{itemize}

To solve these problems, we modify the construction of Theorem \ref{thm:hard}.
The negotiation of Theorem \ref{thm:hard}, which we now call
the old negotiation, just simulates a run of the machine.
The new negotiation also does that in a first stage, but it also exhibits additional behavior. 
Intuitively, if the run rejects, then the new negotiation 
enters a ``anything is possible''- mode, or \emph{anything-mode} for short: in the 
negotiation any sequence of non-final atom may occur, or the 
final atom may occur and the negotiation will terminate. Furthermore, if the run accepts, then in the new negotiation either $n_f$ may occur directly, or the negotiation enters the anything-mode. 

The visual representation of the construction (with a simplified representation of the TM) can be found in Figure \ref{fig:soundExp}. Dashed edges always mean that all agents $C_k$ (in the case of $n_1$) resp. $I$, $P$ and $C_k$ (in the case of $n_2$) have nondeterministic edges for every outcome to that port/from that port to all other ports. \\
\begin{figure}[ht]
\centering
\input{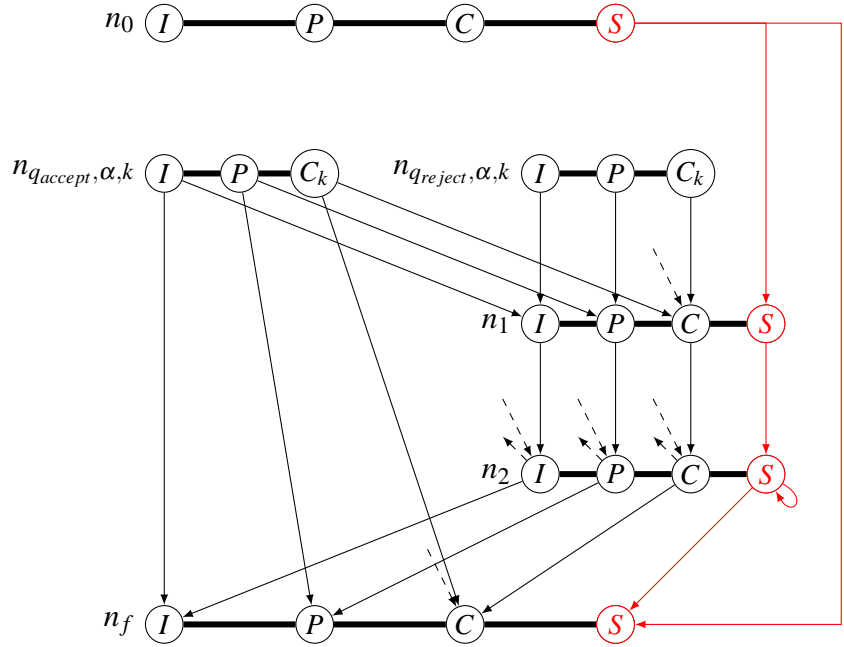}
\caption{Constructing a sound negotiation for the TM reduction. $C$ means all agents $C_k$}
\label{fig:soundExp}
\end{figure}

Now we can argue why the new negotiation is sound. Since, by assumption, 
every run of the machine either accepts or rejects, the negotiation can
always reach the anything-mode, which allows any atom to occur, including
the final atom. 

However, we also have to guarantee that Player 1 has a winning strategy in the new negotiation
if{}f the machine accepts. For this, we ensure two things. First, when the run accepts,
the choice between firing $n_f$ and entering the anything-mode is
taken by Player 1. The winning strategy for Player 1 is to choose $n_f$. Second, in the
anything-mode the choice whether the final atom or some other
atom occurs is controlled by Player 2; therefore, if the
anything-mode is entered after a reject run, Player 2 wins by 
never selecting the final
atom. 

We implement this idea by introducing a new agent $S$ and two new atoms $n_1, n_2$. 
All agents, including $S$, participate in both $n_1$ and $n_2$, but $n_0, n_1, n_2, n_f$
are the only atoms in which $S$ participates. After $n_0$ occurs, agent 
$S$ is ready to engage in both $n_f$ and $n_1$. Firing $n_1$ signals the start of the
anything-mode. After firing $n_1$, all agents move to $n_2$.  Atom $n_2$ has two outcomes,
$s$ and $f$: after the outcome $f$ all agents move to $n_f$, and so by choosing this
outcome the negotiation terminates. After the outcome $s$, all agents but $S$ are ready to
engage in any atom, while $S$ is only ready to engage in $n_2$ again. So this outcome
corresponds to choosing any non-final atom of the old negotiation and firing it. 
But we still have to guarantee that after this atom occurs, atom $n_2$ is enabled
again. For this, we add $n_2$ as $n_2$ to $\trans(n,a,\r)$ for every atom $n \neq \{n_1, n_2, n_f\}$, 
agent $a \neq S$, and outcome $\r$.

To ensure that the anything-mode is entered if the run of the Turing machine rejects,
atoms of the form $n_{q_{reject}, \alpha, k}$ are given one single outcome, 
and the $\trans()$ function is modified as follows: $P$ and $I$ move to $n_1$, and 
we add $n_1$ to the set $\trans(n, C_k, \r)$ for every atom $n$ and outcome $\r$. So
after $n_{q_{reject}, \alpha, k}$ occurs $n_1$ is the only atom enabled.

Now we ensure that if the run of the Turing machine accepts, then Player 1 can decide whether
to terminate or enter the anything-mode. For this
we let atoms of the form $n_{q_{accept}, \alpha, k}$ to be controlled by Player 1, and
add to them a second outcome that enables $n_1$ (additionally to the one enabling $n_f$).

Formally, here are the changes to the construction of Theorem \ref{thm:hard}:

\begin{itemize}
\item $S$ is added to the set of agents: $A' = A \cup \{S\}$
\item $n_1,n_2$, controlled by Player 2 are added to the atoms: $N' = N \cup \{n_1, n_2\}$ and $N_2' = N_2 \cup \{n_1, n_2\}$
\item All agents participate in $n_1, n_2$ (and $S$ in $n_0, n_f$): 
\[P'_{n_0} = P'_{n_f} = P'_{n_1} = P'_{n_2} = A'\]
\item The outcomes for $n_1$ and $n_2$ are defined by: 
$R'_{n_1} = \{\text{s}\}$,
$R'_{n_2} = \{\text{s,f}\}$
\item All atoms $n_{q_{accept},\alpha,k}$ get a second outcome: $R'_{n_{q_{accept},\alpha,k}} = \{\text{s,f}\}$
\item The $\trans(n_1,\_,\_)$ function for atom $n_1$ is given by 
\[\trans'(n_1,a,s) = \{n_2\} \text{ for all } a \in P_{n_1}\]
\item The $\trans(n_2,\_,\_)$ function for atom $n_2$ is given by 
\[\trans'(n_2,a,\text{s}) = \{n \in N'\backslash\{n_f\} : a \in P_n\} \text{ for all } a \in P_{n_2} \backslash \{S\}\]
\[\trans'(n_2,a,\text{f}) = \{n_f\} \text{ for all } a \in P_{n_2}\]
\item The $\trans(\_, S, \_)$ function for agent $S$ is given by:
\[\trans'(n_0,S,\text{st}) = \{n_f,n_1\}\]
\[\trans'(n_2,S,\text{s}) = \{n_2\}\]
\[\trans'(n_2,S,\text{f}) = \{n_2\}\]
\item The $\trans(n,\_,\_)$ function for all accepting and rejecting atoms $n_{q_{accept}, \alpha, k},n_{q_{reject}, \alpha, k}$ is modified to:
\[\trans'(n_{q_{accept},\alpha,k},a,\text{s}) = \{n_1\} \text{ for all } a \in P'_{n_{q_{accept},\alpha,k}}\]
\[\trans'(n_{q_{accept},\alpha,k},a,\text{f}) = \{n_f\} \text{ for all } a \in P'_{n_{q_{accept},\alpha,k}}\]
\[\trans'(n_{q_{reject},\alpha,k},a,\text{s}) = \{n_1\} \text{ for all } a \in P'_{n_{q_{reject},\alpha,k}}\]
\item And finally, all agents except $S$ are ready to take part in $n_2$ at any time:
\[\trans'(n,a,\r) = \trans'(n,a,\r) \cup \{n_2\} \text{ for all } n \in n \in N'\backslash\{n_f\}, a \in P'_n \backslash\{S\}, \r \in R'_n\]
\end{itemize}

Now all atoms are enabled after some sequence, the end atom can always be reached, and $P_1$ can enforce the end atom iff $M$ terminates in $q_{accept}$.

\end{proof}

\end{document}